\documentclass[a4paper,UKenglish,cleveref, autoref, thm-restate]{lipics-v2021}
\pdfoutput=1
\hideLIPIcs
\nolinenumbers
\usepackage{xspace}
\usepackage{dsfont}
\usepackage{amsmath,mathtools}

\let\leftold\left
\let\rightold\right
\renewcommand{\left}{\mathopen{}\mathclose\bgroup\leftold}
\renewcommand{\right}{\aftergroup\egroup\rightold}

\title{Tree Coloring: Random Order and Predictions}

\author{Fabian Frei}{CISPA Helmholtz Center for Information Security, Germany}{fabian.frei@cispa.de}{https://orcid.org/0000-0002-1368-3205}{}
\author{Matthias Gehnen}{Department of Computer Science, RWTH Aachen University, Germany}{gehnen@cs.rwth-aachen.de}{}{}
\author{Dennis Komm}{Department of Computer Science, ETH Zurich, Switzerland}{dennis.komm@inf.ethz.ch}{https://orcid.org/0000-0002-9024-1558}{}
\author{Rastislav Kr\'alovi\v{c}}{Department of Computer Science, Comenius University, Slovakia}{rastislav.kralovic@fmph.uniba.sk}{https://orcid.org/0000-0003-1121-1009}{}
\author{Richard Kr\'alovi\v{c}}{Department of Computer Science, ETH Zurich, Switzerland}{richard.kralovic@inf.ethz.ch}{}{}
\author{Peter Rossmanith}{Department of Computer Science, RWTH Aachen University, Germany}{peter.rossmanith@rwth-aachen.de}{}{}
\author{Moritz Stocker}{Department of Computer Science, ETH Zurich, Switzerland}{moritz.stocker@inf.ethz.ch}{https://orcid.org/0009-0003-8754-5301}{}

\authorrunning{Frei, Gehnen, Komm, Kr\'alovi\v{c}, Kr\'alovi\v{c}, Rossmanith, Stocker}

\Copyright{CC-BY}



\relatedversion{}

\newcommand{\algFF}{\ensuremath{\textup{\textsc{FirstFit}}}\xspace}
\newcommand{\algCBIP}{\ensuremath{\textup{\textsc{CBip}}}\xspace}
\newcommand{\algPFF}{\ensuremath{\textup{\textsc{ParityFirstFit}}}\xspace}
\newcommand{\Prob}[1]{\mathop{\mathrm{Prob}}\left[ #1 \right]\xspace}
\newcommand{\E}[1]{\mathds{E}\left[#1 \right]\xspace}

\newcommand{\euler}{\mathrm{e}\xspace}

\newcommand{\algf}{\(\normalfont{\textup{\textsc{AdviceFirstFit}}}\)\xspace}
\newcommand{\algb}{\(\normalfont{\textup{\textsc{AdviceCBip}}}\)\xspace}
\newcommand{\alg}{\(\normalfont{\textup{\textsc{Alg}}}\)\xspace}
\newcommand{\eps}{\varepsilon}

\DeclarePairedDelimiter{\abs}{\lvert}{\rvert}

\begin{document}

\maketitle

\begin{abstract}
  Coloring is a notoriously hard problem, and even more so in the 
  online setting, where each arriving vertex has to be colored 
  immediately and irrevocably. 
  Already on trees, which are trivially two-colorable, it is 
  impossible to 
  achieve anything better than a logarithmic competitive ratio. 
  We show how to undercut this bound by a double-logarithmic factor 
  in the slightly relaxed online model where the vertices arrive in 
  random order. We then also analyze algorithms with predictions, 
  showing 
  how well we can color trees with machine-learned advice of varying 
  reliability. We further extend our analysis to all two-colorable 
  graphs and provide matching lower bounds in both cases. 
  Finally, we demonstrate how the two mentioned approaches, both of 
  which diminish the often unjustified pessimism of the classical 
  online model, can be combined to yield even better results. 
\end{abstract}

\section{Introduction}

We consider the well-known \emph{online coloring} problem. 
A simple graph (i.e., a graph without weights, orientation, or 
multi-edges) is presented to an algorithm vertex 
by vertex, and each edge is revealed as soon as both endpoints are 
revealed. Both the input graph and its order $n$ are unknown to the 
algorithm, which must assign to each presented vertex a color, 
immediately and irrevocably. As usual for coloring, it is not allowed 
for neighboring vertices to have the same color. The goal is to 
minimize the number of colors used by the algorithm.

Computing the \emph{competitive ratio}~\cite{BE1998} is the classical way
to analyze online algorithms. The competitive ratio is essentially 
the approximation ratio in an online setting; it measures the 
ratio between the cost of the online algorithm's solution and an 
ideal solution that could have been computed had the entire input 
been known from the beginning. A competitive analysis is inherently a 
worst-case analysis. We can thus imagine the graph being revealed to 
the algorithm by
an adversary that tries to maximize the number of colors used by the 
algorithm, which in turn tries to minimize it.

A useful albeit slightly outdated survey on online coloring was 
written by Kierstead~\cite{Kie1998}, a more recent one is found in 
Komm's textbook~\cite{Komm2016}.

Throughout this paper, $\log$ denotes the binary logarithm,  and we identify
colors with the natural numbers $\mathds{N}=\{1,2,3,\dots\}$.

All bipartite graphs, and thus all trees, can be colored with
two colors. (Bipartite graphs are exactly the $2$-colorable graphs, 
in fact.) This is of course easy if the complete instance is 
available but far from feasible in the online setting. It has been 
long known, for example, that deterministic online algorithms can be 
forced to use
at least $\log n+1$ colors, even when instances are restricted to trees;
for more details, see the papers be Bean~\cite{Bea1976} and Gy\'arf\'as and Lehel~\cite{GL1988}.
This yields a lower bound of $(\log n+1)/2$ on the competitive ratio 
of any deterministic online algorithm.
In 1989, Lov\'asz et al.~\cite{LST1989} presented a deterministic 
algorithm, referred to as \algCBIP, 
that uses $2 \log n$ colors on bipartite graphs.  There is an almost 
matching lower bound
of $2\log n - 10$ due to Gutowski et al.~\cite{GKMZ2014}.

Arguably the most straightforward online algorithm to color any kind 
of graph
is \algFF, which implements the simply greedy first-fit strategy of 
assigning the smallest color to the presented
vertex that maintains a valid coloring.  For bipartite graphs, this 
strategy is exponentially worse than that
of \algCBIP, but both algorithms perform equally well on trees.  
Moreover,
Li et al.\ \cite{LNP2022} presented a new model, in which the number of connected
components of an adversarial graph must be bounded after the reveal of each vertex. 
In this model, \algFF even outperforms \algCBIP on trees.

In this paper, we study coloring under two different lenses.

First, we consider the
\emph{random-order model}, where the power of the adversary is significantly 
restricted in that the vertices are revealed to the algorithm in 
random order. In other words, while the adversary decides which graph 
is presented, the vertex order is not under its control but 
chosen uniformly at random. More details about the random-order model can, e.g., be
found in Chapter 11, written by Gupta and Singla \cite{GS2021}, of the
textbook ``Beyond the Worst-Case Analysis of Algorithms.\!''
We show that
\algFF performs significantly better in this model, using
$\mathcal{O}(\log n/(\log\log n))$ colors in expectation. A similar setting of adversarially chosen
instances with randomized presentation is also investigated by 
Burjons et al.~\cite{BHMU2016}. In their model, the adversary
prepares a number of hard instances that are all known to the 
algorithm, but the presented instance is then selected randomly
from the set of instances.

Second, we analyze online coloring by 
algorithms with \emph{predictions}. Representing one approach to 
analyze the power of machine learning, the prediction model has 
garnered plenty of attention recently. Applied to coloring, the idea 
is that the algorithm receives together 
with each vertex a prediction stating to which color this vertex has 
in a fixed optimal solution. As known from the behaviour of artifical 
intelligence, these predictions do not necessarily have to be 
correct, however. Three desirable properties of algorithms dealing 
with this uncertainty are \emph{consistency}, \emph{robustness}, and 
\emph{smoothness}: An 
algorithm is \emph{consistent} if it is optimal for completely 
correct predictions, it is robust if it yields a decent competitive 
ratio even for predictions that are completely wrong, and it is 
\emph{smooth} if the competitive ratio always increases with the 
prediction quality.
The prediction model (also known as the model of machine-learned or 
untrusted advice) was introduced by Lykouris and Vassilvitskii 
\cite{LV2018} and Purohit et al. \cite{PurohitSK18} in 2018 and 
generalized by Angelopoulos et al. \cite{ADJKR2020}. It is based on 
the concept of online \emph{algorithms with advice}, which was introduced a 
decade earlier \cite{BKKKM2009,DKP2008,EFKR2009}, and has also been applied to 
coloring~\cite{BianchiBHK14,FKS2012,SSU2013}. There is a useful 
survey on algorithms with advice by Boyar et al.~\cite{BFKLM2017}. More recently, Antoniadis et al.\ \cite{ABM2024} analyzed the \textsc{FirstFit} algorithm in this model on general graphs. The approach we use is specifically oriented towards trees and, more generally, bipartite graphs, for which we can show stronger results.

After analyzing the two models of random order (\cref{sec:random}) and predictions (\cref{sec:pred})
separately, we finally combine the two approaches in \cref{sec:predrandom} by investigating  
coloring if the vertices of an adversarially chosen instance
are presented in random order but still come with a prediction on what 
color should be chosen.

\section{The Random-Order Model}\label{sec:random}

In the random-order model, the order in which the vertices of the adversarially chosen 
tree $T$ are
revealed is chosen uniformly at random, which gives quite some advantage to the algorithm. 
The simple algorithm $\algFF$, which colors a vertex with the 
smallest color that is not already taken by one of its neighbors, uses at most $\log n+1$ colors on a tree if the order of the 
vertices is adversarially chosen. Since any tree can be colored with two
colors, this corresponds to a competitive ratio of at most $(\log n+1)/2$. In
this section, we will provide an upper bound of $\mathcal{O}(\log n/(\log\log
n))$ for the competitive ratio of \algFF in the random-order model.

We start with two technical lemmata.
\begin{lemma}\label{lem1}
  Let $T=(V,E)$ be any tree. Let $v\in V$ be a vertex of $T$ and let $c(v)$ be the color assigned to $v$ by \algFF.
  Then
  \[ \forall \ell\colon \Prob{\exists v\colon c(v) \ge \ell} \le \frac{n^2}{\ell!}\;. \]
\end{lemma}

\begin{proof}
  We first show that $\Prob{c(v) \ge \ell} \le n/\ell!$ for all $v\in V$ and $\ell\ge 1$.
  The statement of the lemma then follows from a standard union bound over all vertices. 
  To show this, we consider a fixed vertex $v$ with $c(v)\ge \ell$. 
  For each edge $(x, y)$ of $T$ we assign an orientation $x\rightarrow y$ if $x$
  arrived earlier than $y$. We denote by $T_v$ the subtree of $T$ consisting of edges pointing towards $v$. 
  \begin{claim}\label{thm:claim2}
  There is a path in $T_v$ ending in $v$ that contains at least $\ell$ vertices.
  \end{claim}
  \begin{proof}
  We prove this by induction on $\ell$. For $\ell=1$, the path $(v)$ is clearly
  sufficient. Now let $\ell>1$ and assume that the claim holds for all $1\leq \ell'<\ell$. If $c(v)=\ell$, $v$ must be connected to a vertex $v'$ of
  color $\ell-1$ that arrived previously, so the edge $(v'\rightarrow v)$ is
  contained in $T_v$ and $T_{v'}$ is a subtree of $T_v$ that does not contain
  $v$. By the induction hypothesis, there is a path $(v_1,\dots,v_{\ell-1}=v')$ in
  $T_{v'}$, which can be extended to the path $(v_1,\dots,v',v)$ containing $\ell$
  vertices in $T_v$.  This proves \cref{thm:claim2}.  
  \end{proof}
  Now consider any vertex at the beginning of such a path, so any vertex $w$ of $T$ that has a distance from $v$ of at least  $\ell-1$.  Let
  $(w=w_1, w_2, \ldots, w_{\ell'}=v)$ be the (non-oriented) path in $T$ from $w$ to $v$ containing $\ell'\geq \ell$ vertices. Vertex $w$ belongs to $T_v$ only if all the edges of this path are
  oriented towards $v$. There is a single order in which the vertices of this
  path must arrive; and there are $\ell'!$ possible orderings of these vertices, each
  occurring with the same probability. Thus, the probability that $w$ belongs to
  $T_v$ is $1/\ell'!\le 1/\ell!$ and since there are at most $n$ such vertices $w$, we can apply a union bound to show that 
  $\Prob{c(v) \ge \ell} \le n/\ell!$ for all $v\in V$.
\end{proof}

\begin{lemma}\label{lem2}
  Let $c\ge \euler$ with $\euler$ being Euler's number, and let 
  $\ell \ge c\cdot (\log n)/(\log\log n)$. Then
  \[ \ell! \ge n^{c\cdot\left(1-\frac{\log\euler}{\euler}\right)} \ge n^{0.469\cdot c}\;. \]
\end{lemma}

\begin{proof}
  By Stirling's formula, 

  $\ell!\ge 2^{\ell\cdot\log(\ell/\euler)}$. Our conditions on $c$ and $\ell$ imply that

  \[ \ell\cdot\log\left(\frac{\ell}{\euler}\right) \ge c\cdot\log n \cdot\left(1- \frac{\log\log\log n}{\log\log n}\right)\;,\]
  and since $(\log\log\log n)/(\log\log n)$ is bounded from above by $(\log \euler)/\euler$, it follows that
  \[ \ell!\ge 2^{c\cdot\log n \cdot\left(1- \frac{\log \euler}{\euler}\right)} = n^{c\cdot\left(1- \frac{\log \euler}{\euler}\right)}\;. \]
\end{proof}

\begin{theorem}\label{thm:det_up}
  Let $c \ge 2\euler/(\euler-\log \euler) \approx 4.262$.
  The expected number of colors used by \algFF in the random-order 
  model on trees is at most
  \[ c\cdot \frac{\log n}{\log\log n} + 3\;. \]
\end{theorem}

\begin{proof}
  Let $X$ denote a random variable that corresponds to the number of colors used, and let
  $s:= \lceil(c\cdot \log n)/(\log\log n)\rceil$.
  The expected number of colors used is
  \begin{align*}
    \E{X} &= \sum_{\ell\ge 1} \ell\cdot\Prob{X=\ell} = \sum_{1\le \ell \le s}\ell \cdot\Prob{X=\ell} + \sum_{\ell > s}\ell \cdot\Prob{X=\ell}\\
          &\le s + \sum_{\ell> s}\ell \cdot\Prob{X=\ell}\,.
  \end{align*}
  We state the following simple claim.
  \begin{claim}\label{claimfactorial}
 	For any $s\geq 1$, 
 	\[\sum_{\ell\geq s} \frac{1}{\ell!}\leq\frac{2}{s!}\;.\]
  \end{claim}
  \begin{proof}
  We have 
  \[\sum_{\ell\geq s}\frac{1}{\ell!}=\frac{1}{s!}\cdot\sum_{\ell\geq 0} \frac{s!}{(\ell+s)!}=\frac{1}{s!}\cdot\sum_{\ell\geq 0}\prod _{i=1}^\ell \frac{1}{i+s}\leq \frac{1}{s!}\cdot\sum_{\ell\geq 0}\frac{1}{2^{\ell}}=\frac{2}{s!}\;, \]
  which proves \cref{claimfactorial}.
  \end{proof}

  Since $X=\ell$ implies that there is a vertex $v$ with $c(v)\ge \ell$, we 
  can bound the probability $\Prob{X=\ell} \le \Prob{\exists v\colon 
  c(v)\ge \ell}$. This in turn is bounded by \cref{lem1}, which, 
  together with \cref{claimfactorial}, implies that
  \begin{align*}
    \E{X} &\le s + \sum_{\ell > s} \ell\cdot \frac{n^2}{\ell!} = s + n^2 \cdot \sum_{\ell\ge s} \frac{1}{\ell!} \leq s + \frac{2n^2}{s!} \;.
  \end{align*}
  Finally, we can apply \cref{lem2} and get
  \[ \E{X} \le s + 2 n^{2-c\cdot\left(1-\frac{\log\euler}{\euler}\right)}\le s + 2\le c\cdot\frac{\log n}{\log\log n} + 3 \]
  by our choice of $s$.
\end{proof}

\section{Algorithms with Predictions}\label{sec:pred}

In this section, we study online coloring of trees in the 
\emph{predictions model}, which also reduces the overwhelming power 
of the adversary in the classical worst-case model, just like the 
random-order model already analyzed in \cref{sec:random}. However, it does 
so in a completely orthogonal way. 
In the predictions model, there is no random order anymore; the 
adversary fully controls the order in which vertices are revealed. 
The algorithm does receive, however, from an oracle a powerful advice 
bit (a ``prediction'') along with each vertex revealed. This bit tells the algorithm which 
color the revealed vertex has in a fixed $2$-coloring of the tree. 
This in turn allows the algorithm to reconstruct an optimal solution, of 
course, as long as the advice is flawless (this is the property 
commonly called \emph{consistency} for prediction algorithms). The main 
goal of the 
predictions model is to examine what happens if the provided string 
of advice bits contains errors due to inaccurate predictions, for 
example if the advice is provided by some machine learning algorithm.

\subsection{Bounds for Trees with Predictions}

We restate our prediction model more formally. An unknown tree $T$ is 
revealed to an algorithm vertex by vertex. Say that it is a tree on 
the $n$ vertices $v_1,\dots,v_n$, revealed in this order ($n$ is 
not known to the algorithm). The edges are revealed immediately when 
both endpoints have been revealed. Alongside each revealed vertex 
$v_i$, the algorithm receives an advice bit $p(v_i)$, which is 
determined as follows. An oracle chooses a fixed optimal coloring for $T$ 
(that is, a $2$-coloring) and sets 
$p^\ast(v_i)=1$ if $v_i$ has the first 
of the two colors, and $p^\ast(v_i)=0$ 
otherwise. If the prediction is 
correct (that is, does not contain \emph{any} errors), the algorithm receives the advice bits 
$p(v_i)=p^*(v_i)$ for all 
$i\in\{1,\dots,n\}$. But there might 
be, for 
some parameter $k$, up to $k$ indices from $\{1,\dots,n\}$ such that 
$p(v_i)\neq p^\ast(v_i)$. 

\subsubsection{An Upper Bound for Trees with Prediction}
Consider the algorithm \algf that gradually chooses a 
coloring $c\colon T\to\mathds{N}$ (where 
$\mathds{N}=\{1,2,3,\dots\}$) as follows. If the newly revealed 
vertex $v_i$ is currently 
isolated in the graph $T[v_1,\dots,v_i]$ revealed so far, then the 
algorithm completely relies on the advice and colors $v_i$ according 
to the parity of the 
delivered bit; that is, it assigns 
$c(v_i)=1$ if $p(v_i)=1$ and 
$c(v_i)=2$ if $p(v_i)=0$. 
In the other case, that is, if $v_i$ is connected to an already 
revealed vertex, then 
the algorithm \emph{ignores} the advice bit and considers instead the 
revealed neighborhood of $v_i$, which we may 
denote by $N_i=N(v_i)\cap\{v_1,\dots,v_{i-1}\}$, and assigns to $v_i$ 
the lowest remaining color, that is, $c(v_i)=\min(\mathds{N}\setminus 
\{c(v)\mid\,v\in N_i\}$). 

\begin{theorem}
  For $k=0$ (that is, error-free advice), \algf colors the given tree 
  optimally. 
  For $k\ge 1$ errors, it uses at most $\log k+3$ colors. 
  Independent of $k$, it uses at most 
  $\log n+3-\log 3\approx \log n+1.415$ 
  colors on a tree with $n$ vertices.
\end{theorem}

\begin{proof}
We consider any fixed 2-coloring of the tree and fix $p^*(v)$ for all vertices $v$ according to this coloring. Denote by $I(\ell)$ the minimal number of advice bits that \algf 
\emph{ignores} (because it comes along with a vertex neighboring an 
already revealed vertex) while coloring any node of any tree with the color $\ell\geq 1$. Analogously, denote by $F(\ell)$ the minimum number of advice 
bits that are not ignored (because the revealed vertex is currently 
isolated) but \emph{faulty} (i.e., $p(v)\neq p^*(v)$) when \algf colors a node of a tree with the color $\ell$. Since any nontrivial tree requires 
at least two colors even with perfect advice, we have 
$F(1)=F(2)=I(1)=I(2)=0$. Since perfect advice 
always results in a $2$-coloring, we have $F(3)\geq 1$. And since the 
algorithm only uses colors greater than $2$ if it ignores at least one advice 
bit, we have $I(3)\geq 
1$. We now prove that $F(\ell)\geq 2^{\ell-3}$ and $I(\ell)\geq 
2^{\ell-3}$ for $\ell\geq 3$ by induction over $\ell$. 

Let $\ell\geq 4$ and assume by induction that $I(\ell')\geq 2^{\ell'-3}$ and $F(\ell')\geq 2^{\ell'-3}$ for any $3\leq \ell'<\ell$. Assume that \algf colors at least one vertex of 
some given tree with 
the color $\ell$. Let $v_i$ be the first such node. Its neighborhood 
must contain vertices $w_1,w_2,\dots,w_{\ell-1}$ that have previously 
received the 
colors $1,2,\dots,\ell-1$. Among all connected components of the 
subgraph $T[v_1,\dots,v_{i-1}]$ that are already revealed right 
before the $v_i$ appears, denote one containing $w_j$ by $T_{j}$. 
Each connected component $T_j$ is of course still a tree, and each 
$w_j$ is in its own component: because $w_{\ell}$ is adjacent to all 
of them, having two neighbors in the same connected component would 
create a cycle. 
It follows that separate advice bits are provided alongside the 
vertices of each component. Since $T_j$ contains a vertex of color 
$j$, among the advice bits delivered together with the vertices of
$T_j$ there 
must be at least $I(j)$ ignored ones and $F(j)$ heeded but faulty ones.

Additionally, we observe that the color of either $w_1$ or 
$w_2$ is wrong because they are both neighbors of $v_i$ in the 
completed tree and thus any optimal solution uses the same color 
for $w_1$ and $w_2$. This implies that at least one faulty bit was 
used for a vertex in $T_1$ (the component of $w_1$) or $T_2$ (the 
component of $w_2$), which eventually led to the wrong color for 
$w_1$ or $w_2$. It follows that the number of faulty but heeded bits in $T$ is at least
\[F(\ell)\geq F(\ell-1)+\dots+F(3)+1\;.\]
We also know that the advice bit for $v_i$ is ignored because it is 
colored by $\ell>2$, which implies 
\[I(\ell)\geq I(\ell-1)+\dots+I(3)+1\;.\]
By induction, we obtain 
\[F(\ell)\geq 2^{\ell-4}+\dots+2^0+1=2^{\ell-3}\;.\]
and 
\[I(\ell)\geq 2^{\ell-4}+\dots+2^0+1=2^{\ell-3}\;.\]
We conclude that for a fixed number of $k$ errors, the 
largest color that \algf uses is at most $\log k+3$. 

Now recall that the analysis on the number of ``faults'' in the advice string was based on a fixed 2-coloring of the tree.
Flipping every bit of a correct advice string for an optimal coloring 
results in 
another such advice string for the same solution with the two colors 
swapped. We can thus assume without loss of generality that
at most half of the bits whose advice was followed to color the tree are 
faulty (by switching the optimal solution that we consider to decide 
which bits are faulty). The number of ignored advice bits remains 
unchanged; it is still equal to the number of vertices that are not
isolated upon their reveal. The total number of vertices in the 
graph is therefore at least $2\cdot F(\ell)+I(\ell)\geq 3\cdot 
2^{\ell-3}$. For a tree on $n$ vertices, the algorithm thus
uses at most $\log n+3-\log 3\leq \log n+1.4151$ colors.
\end{proof}

\subsubsection{A Lower Bound for Trees with Predictions}
We now provide a perfectly matching 
lower bound on the number of 
colors required to color graphs of up to
$n$ vertices with a consistent 
algorithm. 
\begin{theorem}\label{thm:lowerbound}
There is no consistent prediction 
algorithm using fewer than $\log 
n+3-\log 3\approx \log n+1.415$ colors  on trees.
\end{theorem}

\begin{proof}
We will first consider how such an algorithm colors isolated vertices. Consider the case where an algorithm \alg is presented with a 
non-negative number of isolated vertices $v_i^0$ whose advice bit is 
$0$ and a non-negative number of isolated vertices $v_i^1$ whose 
advice bit is~$1$. An optimal coloring uses only two colors. If \alg 
is consistent, it must color the vertices $v_i^0$ with one 
color and all vertices $v_i^1$ with the color. Otherwise, a vertex $w^0$ 
neighboring all vertices $v_i^0$ and a vertex $w^1$ neighboring all 
vertices $v_i^1$ and $w^0$ will be revealed. If two vertices 
neighboring $w^0$ are colored differently, then $w^0$ must use a 
third color. The analogous statement holds for $w^1$. 
The remaining case is that all vertices except for $w^0$ and $w^1$ 
have the same color, implying that $w^0$ and $w^1$ cannot use this 
color. But because they are adjacent, at least two more colors are 
required for a valid coloring. 

Now let $\ell\geq 3$. We prove that we can present a 
tree on $3\cdot 2^{\ell-3}$ vertices to \alg such that it uses at 
least $\ell$ colors. We argue by induction on $\ell$. In the case 
$\ell=3$, we present two isolated vertices $v_1$ and $v_2$ with 
distinct advice bits. By the argument from the previous paragraph, 
\alg must color them with distinct colors. We then present a third 
vertex $w$ connected to $v^0$ and $v^1$, forcing \alg to use a third 
color. For $\ell\geq 4$, we present the algorithm with two isolated 
vertices $v_1$ and $v_2$ with distinct advice bits. By the same 
argument 
again, \alg must color them with distinct colors, which we assume 
to be $1$ and $2$ without loss of generality. We now present the 
algorithm with $(\ell-3)$ disjoint subtrees $T_i$ for 
$i\in\{3,4,\dots,
\ell-1\}$. By our induction hypothesis, we can indeed choose a $T_i$ 
of size $3\cdot 2^{i-3}$ 
that forces \alg to use at least $i$ colors on the vertices of $T_i$. 
We now inductively choose $c_i$ for $i\in\{3,4,\dots,
\ell-1\}$ to be a 
color used by \alg in the coloring of $T_i$ that is not in the set 
$\{1,2,c_3,\dots,c_{i-1}\}$. This is possible since this is a set 
of $i-1$ colors. Let $v_i$ be a vertex of $T_i$ colored with 
$c_i$. We now introduce a final vertex $w$ adjacent to all of $v_1$, 
$v_2$ 
and $v_3,\dots,v_{\ell-1}$. Since the subtrees are disjoint, the 
result is a tree. We know that \alg cannot color $w$ in any of the 
colors 
$\{1,2,c_3,\dots,c_{\ell-1}\}$ and thus uses at least $\ell$ colors 
on the presented tree. The total number of vertices in the tree is 
$3+\abs{T_3}+\dots+\abs{T_{\ell-1}}$, which is
\[3+\sum_{i=3}^{\ell-1} 3\cdot 2^{i-3}=3\cdot (1+\sum_{j=0}^{\ell-4} 2^j)=3\cdot 2^{\ell-3}\;.\]
\end{proof}

\subsection{Generalization to Bipartite Graphs}
In this section we consider the generalization of our problems for 
trees to general bipartite graphs, which is a natural choice since 
these are precisely the two-colorable graphs. The algorithm \algCBIP always
colors a vertex with the smallest color that is not in the opposite partition
of the connected component of that vertex when it is revealed \cite{LST1989}.
In the model without predictions, this algorithm is (asymptotically) the best
possible \cite{GKMZ2014,LST1989}. We will extend \algCBIP to an algorithm
\algb that makes use of predictions.
 
\subsubsection{Upper Bounds for Bipartite Graphs with Predictions}

For upper bounds on the number of colors needed on bipartite graphs, 
we consider the algorithm \algb that receives the vertices 
$v_1,\dots,v_n$ of graph $G$ and an advice bit $p(v_i)$ for each 
vertex $v_i$ one after the other. If the vertex $v_i$ is isolated in 
the graph $G[v_1,\dots,v_i]$, the algorithm colors $v_i$ according to 
the parity of the advice bit, so $c(v_i)=1$ if $p(v_i)=1$, and 
$c(v_i)=2$ if $p(v_i)=0$. If not, it looks at the connected component 
of $G[v_1,\dots,v_i]$ containing $v_i$. This component can be 
partitioned into two independent sets. \algb then colors $v_i$ with 
the smallest color that is not represented in the independent set 
that does not contain~$v_i$.

We now provide two different upper bounds, the first one depending on 
the number of errors in the provided prediction, the second one 
depending on the order of the input graph. 
\begin{theorem}\label{thm:biperror}
  If there are no errors in the prediction, \algb colors bipartite graphs optimally. If there are at most $k\geq 1$ errors, \algb uses at most $2\cdot\log k+4$ colors.
\end{theorem}

\begin{proof}
Let $F(\ell)$ be the number of errors in the advice string of a 
bipartite graph needed for the algorithm \algb to color at least one 
of its vertices with the color $\ell$. Clearly $F(1)=F(2)=0$. We now 
prove that $F(\ell)\geq  \lceil 2^{\ell/2-2}\rceil$ for $\ell\geq 3$ 
by induction on $\ell$. We have $F(3)\geq 1$ and $F(4)\geq 1$ since 
no vertex will be colored by any other color than $1$ or $2$ if the 
advice string is correct (it can be easily checked that in fact 
$F(3)=F(4)=1$). For $\ell \geq 5$, assume that \algb colors at least 
one node with the color $\ell$. Let $v_i=w_\ell$ be the first such 
node. 
Let $S_\ell\cup\bar{S_\ell}$ be a partition of the connected 
component of $G[v_1,\dots,v_i]$ that contains $v_i$, such that 
$v_i\in S_\ell$. There must be vertices 
$\bar{w}_1,\dots,\bar{w}_{\ell-1}$ in $\bar{S}_{\ell}$ such that $\bar{w}_j$ is 
colored with color $j$ for all $1\leq j\leq \ell-1$. Now consider the 
moment when $v_{i'}=\bar{w}_{\ell-1}$, $i'<i$ was revealed and let 
$T_{\ell-1}\cup \bar{T}_{\ell-1}$ be a partition of the connected 
component of $G[v_1,\dots,v_{i'}]$ that contains $v_{i'}$, such that 
$v_{i'}\in \bar{T}_{\ell-1}$. There must be vertices 
$w_1,\dots,w_{\ell-2}$ in $T_{\ell-1}\subseteq S_\ell$ such that 
$w_{j}$ is colored with color $j$ for all $1\leq j\leq \ell-2$. Now 
assume w.l.o.g. that $w_{\ell-2}$ is revealed after 
$\bar{w}_{\ell-2}$. Clearly, when $w_{\ell-2}$ is revealed, 
$w_{\ell-2}$ and $\bar{w}_{\ell-2}$ were in different connected 
components, otherwise they could not have received the same color. 
Hence \algb colored a vertex with color $\ell-2$ in two separate 
bipartite graphs with disjoint advice strings. By induction, there 
must be at least $F(\ell-2)$ errors in each advice string, and hence 
at least $2\cdot F(\ell-2)$ errors in the advice string for $G$. Thus 
\[F(\ell)\geq 2\cdot F(\ell-2)\geq 2\cdot \lceil 
2^{(\ell-2)/2-2}\rceil=2\cdot \lceil 2^{ \ell/2-3}\rceil\geq \lceil 
2^{\ell/2-2}\rceil\;.\]
So for any fixed number of $k$ errors, the largest color that \algb 
uses to color the graph is at most $2\cdot\log k+4$.
\end{proof}

We now provide for the same algorithm \algb an upper bound on the 
number of used colors depending on the number of vertices of the 
input graph. 
\begin{theorem}\label{thm:bipcolors}
On a bipartite graph with $n$ 
vertices, \algb uses at most 
$2\cdot\log (n)-1.64$ colors for any 
$n\ge1500$.
\end{theorem}
\begin{proof}
Denote by $N(\ell)$ the minimal number of vertices a graph must 
contain for \algb to color at least one of its vertices with the 
color $\ell$. 
We will first show that $N(\ell)\geq 5\cdot 2^{(\ell-3)/2}-2$ for odd $\ell\geq 3$, and that $N(\ell)\geq 2\cdot 2^{\ell/2+1}-2$ for even $\ell\geq 4$. 
Clearly, $N(3)\geq 3=5\cdot 2^{(3-3)/2}-2$, since \algb 
will always color the second vertex of a graph with either $1$ or 
$2$. We also have $N(4)\geq 6=2^{4/2+1}-2$ because for any vertex to 
be assigned color $4$, there must be vertices of color $1$, $2$ and 
$3$ on the opposite shore of the connected component, which in turn 
implies that there are vertices of color $1$ and $2$ on the shore 
where to vertex of color $4$ lies, a total of at least $6$ vertices. 
We now argue by induction on $\ell$. Assume that \algb colors a 
vertex of some bipartite graph with the color $\ell\geq 5$. As seen 
in the proof of \cref{thm:biperror}, \algb must have colored two 
disjoint bipartite graphs with at least $\ell-2$ colors each, and 
additionally there must be at least one vertex of color $\ell-1$ and 
on vertex of color $\ell$. This implies that 
\[N(\ell)\geq 2\cdot N(\ell-2)+2\;.\]
By induction, we have for even $\ell$, \[N(\ell)\geq 2\cdot 
(2^{(\ell-2)/2+1}-2)+2=2\cdot (2^{\ell/2}-2)+2=2^{\ell/2+1}-2\;,\]
and for odd $\ell$,
\[N(\ell)\geq 2\cdot (5\cdot 2^{(\ell-5)/2}-2)+2=5\cdot 2^{(\ell-3)/2}-4+2=5\cdot 2^{(\ell-3)/2}-2\;.\]
Since $5\cdot 2^{(\ell-3)/2}=(5/\sqrt{8})\cdot 2^{\ell/2}<2\cdot 
2^{\ell/2}=2^{\ell/2+1}$, we can further conclude that for all $\ell\geq 3$, $N(\ell)\geq 
5\cdot 2^{(\ell-3)/2}-2$. In particular, for a 
fixed $n$, \algb will color any bipartite graph on $n$ vertices using 
at most $2\cdot\log(n+2)+3-2\cdot\log(5)$ colors, which is less than  
$2\cdot\log n-1.64$ for all $n\ge 1500$.
\end{proof}

\subsubsection{Lower Bounds for Bipartite Graphs with Predictions}

We mainly recall here the lower bound of $2\log n - 10$ colors for 
bipartite graphs due to Gutowski et al.~\cite{GKMZ2014}, 
which neatly complemented the upper bound of the $\log n$-competitive 
deterministic algorithm \algCBIP by Lov\'asz et al.~\cite{LST1989}. 
In the following lemma, we improve the upper bound slightly by an 
improved analysis for \algCBIP 
and show that a consistent algorithm with predictions cannot beat 
\algCBIP on all instances. 

\begin{proposition}
The algorithm \algCBIP uses at most $2\cdot\log n-1.999$ colors on 
graphs with $n\ge5770$ vertices, and a consistent algorithm with 
predictions cannot perform better than \algCBIP on all 
graphs. 
\end{proposition}

\begin{proof}
For the deterministic algorithm $\textsc{CBip}$ we can see as in the 
proof of \cref{thm:bipcolors} that for graphs to be colored in 
$\ell$ colors, $2^{\ell/2+1}-2$ vertices are needed if $\ell$ is even 
and $3\cdot 2^{(\ell-1)/2}-2$ vertices are needed if $\ell$ is odd, 
so at least $2^{\ell/2+1}-2$ vertices are needed in either case. For 
a fixed number $n$ of vertices, this implies that $\textsc{CBip}$ 
requires at most $2\cdot \log(n+2)-2$ colors, 
which is at most $2\cdot\log n-1.999$ for $n\ge5770$ and $2\cdot\log 
n-2+\eps$ for any $\eps>0$ and $n$ large enough.
To show that a consistent prediction algorithm cannot beat \algCBIP, 
it suffices to show with the argument from the proof of 
\cref{thm:lowerbound} that any consistent algorithm can be forced 
to use three colors on an instance with three vertices, while  
$\textsc{CBip}$ only ever uses two colors on such instances, 
independent of the order in which the vertices are presented.
\end{proof}

\section{Algorithms with Predictions in the Random-Order Model}\label{sec:predrandom}

We now combine the two ways of diminishing the overly 
pessimistic assumptions of the classical online coloring model by 
considering an adversary that has no control over the order in which 
vertices arrive, and an algorithm that receives predictions of the 
type described in \cref{sec:pred}. 

Let $T$ be an arbitrary tree. We have a prediction function that assigns a
1-bit prediction $p(v)$ to every vertex $v$ of $T$. We fix consider a fixed $2$-coloring of $T$ and denote by $p^*(v)\in \{0,
1\}$ the parity of the color of $v$ in that coloring. We assume that $p(v)$
approximates $p^*(v)$ with $k$ errors.
We assume that vertices of $T$ arrive in random order together with their
predictions. The predictions thus do not depend on the random order of vertices.
We analyze the algorithm $\algPFF$, which colors a vertex $v$ with 
the smallest possible color that has the same parity as
$p(v)$. Let $c(v)$ be the color assigned to $v$ by $\algPFF$.

As in \cref{sec:random}, for each edge $(x,y)$ of $T$ we assign an orientation $x\rightarrow y$
if $x$ arrived earlier than $y$ and denote by $T_v$ the subtree of $T$ consisting of edges pointing towards $v$. 

\begin{lemma}\label{lmErrorsToColors}
Consider any $T_v$ such that $\algPFF$ colors $v$ with color $\ell$.
Then there exists a path in $T_v$ ending in $v$ that contains at least $\lfloor(\ell-1)/4\rfloor$ vertices with incorrect predictions.
\end{lemma}

\begin{proof}
We prove this claim by induction on $\ell$. For $\ell\leq 4$, we have $\lfloor
(\ell-1)/4\rfloor=0$, so the path $(v)$ is sufficient. Now assume by induction that the claim holds for $1\leq \ell'<\ell$ and suppose a vertex $v$ is
colored with $c(v)=\ell\geq 5$. Then $v$ must be connected to a previously
revealed vertex $v'$ with $c(v')=\ell-2$. Since $v$ and $v'$ are connected and
are colored with the same parity, either $p(v)\neq p^*(v)$ or $p(v')\neq
p^*(v')$. The vertex $v'$ in turn must be connected to a vertex $v''\neq v$
with $c(v'')=\ell-4$ and $T_{v''}$ is a subtree of $T_v$. By induction, there
must be a path containing at least $\lfloor (\ell-5)/4\rfloor$ vertices with
incorrect predictions in $T_{v''}$ ending in $v''$. By extending this path with
the vertices $v'$ and $v$, we can construct a path in $T_v$ that contains at
least $\lfloor(\ell-5)/4\rfloor+1=\lfloor (\ell-1)/4\rfloor$ vertices with
incorrect predictions ending in $v$.
\end{proof}

\begin{lemma}\label{lem3}
For any color $\ell\geq 7$, if there are a total of $k$ vertices in $T$ whose predictions are incorrect, then
\[ \Prob{\exists v: c(v) \ge \ell} \le \frac{k^2}{\left\lfloor\frac{\ell-3}{4}\right\rfloor!}\;. \]
\end{lemma}

\begin{proof}
We first show that for any color $\ell\geq 5$ and for any vertex $v\in V$
\[\Prob{c(v) \ge \ell} \le \frac{k}{\left\lfloor\frac{\ell-1}{4}\right\rfloor!}\,.\]
To prove this, consider a fixed vertex $v$ such that $c(v)\ge \ell\geq 5$. 
Due to \cref{lmErrorsToColors}, there must be a path $(w_1,\dots,w_{r-1},v)$ in
$T_v$ containing at least $\lfloor (\ell-1)/4\rfloor$ vertices whose predictions
are incorrect. By shortening the path from the front, we can assume that
$p(w_1)\neq p^*(w_1)$.

Now consider any vertex $w_1$ whose path to $v$ in $T$ contains $r\geq \lfloor
(\ell-1)/4\rfloor$ vertices. Then $w_1$ is in $T_v$ if and only if the vertices
$w_1,\dots,w_{r-1},v$ arrived in order. Since every order of these vertices is
equally likely, we have
\[ \Prob{w\in T_v}=\frac{1}{r!}\le \frac{1}{\left\lfloor\frac{\ell-1}{4}\right\rfloor!}\;. \]
By a union bound over all vertices $w_1$ of $T$ whose predictions are incorrect, we can see that
\[\Prob{c(v) \ge \ell} \le \frac{k}{\left\lfloor\frac{\ell-1}{4}\right\rfloor!}\;.\]
To prove the lemma, note that if there is a vertex $v\in T$ with $c(v)\ge
\ell\geq 7$, then $v$ is connected to a vertex $v'$ with $c(v')=\ell-2\geq 5$.
Since $v$ and $v'$ were colored with the same parity, one of them must have
received an incorrect prediction. In either case there is a vertex whose
prediction was incorrect and that was colored with a color at least $\ell-2$.  

Applying a union bound over all vertices $v'\in T$ with incorrect predictions, we get that 
\[ \Prob{\exists v: c(v) \ge \ell} \le \Prob{\exists v': c(v')\ge \ell-2, p(v')\neq p^*(v')} \le \frac{k^2}{\left\lfloor\frac{\ell-3}{4}\right\rfloor!}\;.\]
\end{proof}

\begin{theorem}
Let $c \ge \frac{2\euler}{\euler-\log \euler} \approx 4.262$.
The expected number of colors used by the algorithm $\algPFF$  on 
trees is at most
\[ 4c\cdot \frac{\log k}{\log \log k} + 71\;. \]
\end{theorem}

\begin{proof}
Let $X$ be the number of colors used and let
\[ s:= \left\lceil\frac{4c\cdot \log k}{\log\log k} + 10\right\rceil\;. \]
The expected number of colors used is
\begin{align*}
\E{X} &= \sum_{\ell\ge 1} \ell\cdot\Prob{X=\ell} \\
      &= \sum_{1\le \ell \le s}\ell \cdot\Prob{X=\ell} + \sum_{l > s}l \cdot\Prob{X=\ell} \\
      &\le s + \sum_{\ell > s}\ell \cdot\Prob{X=\ell}\;.
\end{align*}

Note that $X=\ell$ implies that $\exists v: c(v)\ge \ell$. Thus $\Prob{X=\ell} \le \Prob{\exists v: c(v)\ge \ell}$,
which is bounded by \cref{lem3}, so
\begin{align*}
\E{X} &\le s + \sum_{\ell > s} \ell\cdot \frac{k^2}{\left\lfloor\frac{\ell-3}{4}\right\rfloor!} \\
      &=   s + k^2 \cdot \sum_{\ell \ge s+1} \frac{\ell}{\left\lfloor\frac{\ell-3}{4}\right\rfloor!} \\
      &\le   s + k^2 \cdot \sum_{i \ge \left\lfloor \frac{s-2}{4}\right\rfloor} \frac{16i + 14}{i!} \\
      &=   s + 16k^2 \cdot \sum_{i \ge \left\lfloor \frac{s-2}{4}\right\rfloor - 1} \frac{1}{i!} +
               14k^2 \cdot \sum_{i \ge \left\lfloor \frac{s-2}{4}\right\rfloor} \frac{1}{i!} \\
      &\le s + 30k^2 \cdot \sum_{i \ge \left\lfloor \frac{s-2}{4}\right\rfloor -1 } \frac{1}{i!}
\end{align*}
By \cref{claimfactorial} we have
\[ \E{X} \le s + 60k^2 \cdot \frac{1}{\left\lfloor \frac{s-6}{4}\right\rfloor!}\;, \]
and since
\[ \left\lfloor \frac{s-6}{4}\right\rfloor \ge \frac{c\log k}{\log \log k} \]
by our choice of $s$, we can invoke \cref{lem2} to obtain
\[ \left\lfloor \frac{s-6}{4}\right\rfloor! \ge k^{c\cdot\left(1-\frac{\log \euler}{\euler}\right)}\;. \]
Thus,
\[ \E{X} \le s + 60k^{2-c\cdot\left(1-\frac{\log \euler}{\euler}\right)} \]
and by our choice of $c$ we have
\[ \E{X} \le s + 60\le 4c\cdot\frac{\log k}{\log\log k} + 71\;. \]
\end{proof}

\section{Conclusion}

We have analyzed the behavior of the \algFF algorithm on trees in the 
random-order model and showed that its competitive ratio improves from
$\Theta(\log n)$ to $\mathcal{O}(\log n/(\log \log n))$. 
We further looked into algorithms with predictions for the colors of 
the vertices and provided consistent and smooth algorithms with a 
robustness that matches the best algorithms without predictions on 
trees and bipartite graphs, respectively.
Lastly, we combined the random-order model with predictions and gave 
a consistent and smooth algorithm with a robustness of 
$\mathcal{O}(\log n/(\log \log n))$, which matches our algorithm 
without predictions.

The performance of algorithms such as \algFF and \algCBIP on 
bipartite graphs in the random-order model remains as an open 
problem. The most interesting question, however, is whether we can 
prove a matching lower bound for the improved competitive ratio in the 
random-order model.

\bibliography{refs}

\begin{thebibliography}{10}

\bibitem{ADJKR2020}
Spyros Angelopoulos, Christoph D{\"u}rr, Shendan Jin, Shahin Kamali, and Marc
  Renault.
\newblock {Online Computation with Untrusted Advice}.
\newblock In {\em Proceedings of the 11th Innovations in Theoretical Computer
  Science Conference (ITCS 2020)}, volume 151 of {\em Leibniz International
  Proceedings in Informatics (LIPIcs)}, pages 52:1--52:15, 2020.

\bibitem{ABM2024}
Antonios Antoniadis, Hajo Broersma, and Yang Meng.
\newblock Online graph coloring with predictions.
\newblock In Amitabh Basu, Ali~Ridha Mahjoub, and Juan~Jos{\'e}
  Salazar~Gonz{\'a}lez, editors, {\em Combinatorial Optimization (ISCO 2024)},
  pages 289--302. Springer Nature Switzerland, Cham, 2024.
\newblock \href {https://doi.org/10.1007/978-3-031-60924-4_22}
  {\path{doi:10.1007/978-3-031-60924-4_22}}.

\bibitem{Bea1976}
Dwight~R. Bean.
\newblock Effective coloration.
\newblock {\em The Journal of Symbolic Logic}, 41(2):469--480, 1976.

\bibitem{BianchiBHK14}
Maria~Paola Bianchi, Hans{-}Joachim B{\"{o}}ckenhauer, Juraj Hromkovic, and
  Lucia Keller.
\newblock Online coloring of bipartite graphs with and without advice.
\newblock {\em Algorithmica}, 70(1):92--111, 2014.

\bibitem{BKKKM2009}
Hans-Joachim B\"ockenhauer, Dennis Komm, Rastislav Kr\'alovi\v{c}, Richard
  Kr\'alovi\v{c}, and Tobias M\"omke.
\newblock On the advice complexity of online problems.
\newblock In Yingfei Dong, Ding-Zhu Du, and Oscar~H. Ibarra, editors, {\em
  Proceedings of the 20th International Symposium on Algorithms and Computation
  (ISAAC 2009)}, volume 5878 of {\em Lecture Notes in Computer Science}, pages
  331--340. Springer-Verlag, Berlin, 2009.

\bibitem{BE1998}
Allan Borodin and Ran El-Yaniv.
\newblock {\em Online Computation and Competitive Analysis}.
\newblock Cambridge University Press, Cambridge, 1998.

\bibitem{BFKLM2017}
Joan Boyar, Lene~M. Favrholdt, Christian Kudahl, Kim~S. Larsen, and Jesper~W.
  Mikkelsen.
\newblock Online algorithms with advice: {A} survey.
\newblock {\em {ACM} Comput. Surv.}, 50(2):19:1--19:34, 2017.
\newblock \href {https://doi.org/10.1145/3056461} {\path{doi:10.1145/3056461}}.

\bibitem{BHMU2016}
Elisabet Burjons, Juraj Hromkovi\v{c}, Xavier Mu\~{n}oz, and Walter Unger.
\newblock Graph coloring with advice and randomized adversary (extended
  abstract).
\newblock In {\em Proceedings of the 42nd International Conference on Current
  Trends in Theory and Practice of Computer Science (SOFSEM 2016)}, volume 9587
  of {\em Lecture Notes in Computer Science}. Springer-Verlag, Berlin, 2016.
\newblock 229--240.

\bibitem{DKP2008}
Stefan Dobrev, Rastislav Kr\'alovi\v{c}, and Dana Pardubsk\'a.
\newblock How much information about the future is needed?
\newblock In Viliam Geffert, Juhani Karhum\"aki, Alberto Bertoni, Bart Preneel,
  Pavol N\'avrat, and M\'aria Bielikov\'a, editors, {\em Proceedings of the
  34th Conference on Current Trends in Theory and Practice of Computer Science
  (SOFSEM 2008)}, volume 4910 of {\em Lecture Notes in Computer Science}, pages
  247--258. Springer-Verlag, Berlin, 2008.

\bibitem{EFKR2009}
Yuval Emek, Pierre Fraigniaud, Amos Korman, and Adi Ros{\'e}n.
\newblock Online computation with advice.
\newblock In Susanne Albers, Alberto Marchetti-Spaccamela, Yossi Matias,
  Sotiris~E. Nikoletseas, and Wolfgang Thomas, editors, {\em Proceedings of the
  36th International Colloquium on Automata, Languages and Programming (ICALP
  2009)}, volume 5555 of {\em Lecture Notes in Computer Science}, pages
  427--438. Springer-Verlag, Berlin, 2009.

\bibitem{FKS2012}
Michal Fori\v{s}ek, Lucia Keller, and Monika Steinov\'a.
\newblock Advice complexity of online coloring for paths.
\newblock In {\em Proceedings of the 6th International Conference on Language
  and Automata Theory and Applications (LATA 2012)}, volume 7183 of {\em
  Lecture Notes in Computer Science}, pages 228--239. Springer-Verlag, Berlin,
  2012.

\bibitem{GS2021}
Anupam Gupta and Sahil Singla.
\newblock Random-order models.
\newblock In Tim Roughgarden, editor, {\em Beyond the Worst-Case Analysis of
  Algorithms}, pages 234--258. Cambridge University Press, 2021.
\newblock \href {https://doi.org/10.1017/9781108637435.015}
  {\path{doi:10.1017/9781108637435.015}}.

\bibitem{GKMZ2014}
Grzegorz Gutowski, Jakub Kozik, Piotr Micek, and Xuding Zhu.
\newblock Lower bounds for on-line graph colorings.
\newblock In Hee-Kap Ahn and Chan-Su Shin, editors, {\em Proceedings of the
  25th International Symposium on Algorithms and Computation (ISAAC 2014)},
  volume 8889 of {\em Lecture Notes in Computer Science}, pages 507--515.
  Springer-Verlag, Berlin, 2014.

\bibitem{GL1988}
Andr\'as Gy\'arf\'as and Jen\"o Lehel.
\newblock On-line and first fit colorings of graphs.
\newblock {\em Journal of Graph Theory}, 12(2):217--227, 1988.

\bibitem{Kie1998}
Henry~A Kierstead.
\newblock Recursive and on-line graph coloring.
\newblock In {\em Studies in Logic and the Foundations of Mathematics}, volume
  139, pages 1233--1269. Elsevier, 1998.

\bibitem{Komm2016}
Dennis Komm.
\newblock {\em An Introduction to Online Computation -- Determinism,
  Randomization, Advice}.
\newblock Texts in Theoretical Computer Science. An {EATCS} Series. Springer,
  2016.
\newblock \href {https://doi.org/10.1007/978-3-319-42749-2}
  {\path{doi:10.1007/978-3-319-42749-2}}.

\bibitem{LNP2022}
Yaqiao Li, Vishnu~V. Narayan, and Denis Pankratov.
\newblock Online coloring and a new type of adversary for online graph
  problems.
\newblock {\em Algorithmica}, 84(5):1232--1251, 2022.
\newblock \href {https://doi.org/10.1007/s00453-021-00920-w}
  {\path{doi:10.1007/s00453-021-00920-w}}.

\bibitem{LST1989}
L\'aszl\'o Lov\'asz, Michael Saks, and William~T. Trotter.
\newblock An on-line graph coloring algorithm with sublinear performance ratio.
\newblock {\em Discrete Mathematics}, 75(1--3):319--325, 1989.

\bibitem{LV2018}
Thodoris Lykouris and Sergei Vassilvitskii.
\newblock Competitive caching with machine learned advice.
\newblock In {\em Proceedings of the 35th International Conference on Machine
  Learning (ICML 2018)}, volume~80 of {\em Proceedings of Machine Learning
  Research}, pages 3302--3311. {PMLR}, 2018.

\bibitem{PurohitSK18}
Manish Purohit, Zoya Svitkina, and Ravi Kumar.
\newblock Improving online algorithms via {ML} predictions.
\newblock In {\em Advances in Neural Information Processing Systems 31: Annual
  Conference on Neural Information Processing Systems 2018 (NeurIPS 2018)},
  pages 9684--9693, 2018.

\bibitem{SSU2013}
Sebastian Seibert, Andreas Sprock, and Walter Unger.
\newblock Advice complexity of the online coloring problem.
\newblock In {\em Proceedings of the 8th International Conference on Algorithms
  and Complexity (CIAC 2013)}, volume 7878 of {\em Lecture Notes in Computer
  Science}, pages 345--357. Springer-Verlag, Berlin, 2013.

\end{thebibliography}

\appendix

\end{document}